\documentclass[10pt,twocolumn]{IEEEtran}
\usepackage{url}
\usepackage{array}
\usepackage{makecell}
\usepackage[usenames,dvipsnames]{pstricks}
\usepackage{epsfig}
\usepackage{pst-grad} 
\usepackage{pst-plot} 
\usepackage{amsmath,epsfig}
\usepackage{amssymb}
\usepackage{floatrow}
\usepackage{mathtools}

\usepackage{enumitem}
\usepackage{psfrag}
\usepackage{color}
\usepackage{pstricks}
\usepackage{cite}
\usepackage{algorithm}
\usepackage{epsfig}
\usepackage{amsthm}
\usepackage{amssymb}
\usepackage{psfrag}
\usepackage{pstricks}
\usepackage{float}
\usepackage{stfloats}
\usepackage{wrapfig}
\usepackage{graphicx}
\usepackage{amsmath}
\usepackage{amsfonts}
\usepackage{amssymb}
\usepackage{algorithm}
\usepackage{blindtext}
\usepackage{pstricks,pst-plot}
\usepackage{pst-bezier}
\usepackage{pst-math}
\usepackage{algorithm,algcompatible,amsmath}
\usepackage{hyperref}
\usepackage{float}
\usepackage{cite}
\usepackage{amsmath,amssymb,amsfonts,dsfont}
\usepackage{scalerel,stackengine}
\stackMath
\newcommand\reallywidehat[1]{%
\savestack{\tmpbox}{\stretchto{%
  \scaleto{%
    \scalerel*[\widthof{\ensuremath{#1}}]{\kern-.6pt\bigwedge\kern-.6pt}%
    {\rule[-\textheight/2]{1ex}{\textheight}}
  }{\textheight}%
}{0.5ex}}%
\stackon[1pt]{#1}{\tmpbox}%
}
\usepackage{algorithm}
\usepackage[english]{babel}
\usepackage[utf8]{inputenc}
\usepackage[noend]{algpseudocode}
\usepackage{graphicx}
\usepackage{textcomp}
\usepackage{setspace}
\usepackage [english]{babel}
\usepackage [autostyle, english = american]{csquotes}
\usepackage[dvipsnames]{xcolor}
\colorlet{LightRubineRed}{RubineRed!70!}
\colorlet{Mycolor1}{green!10!orange!90!}
\definecolor{Mycolor2}{HTML}{420DAB} 
\definecolor{Mycolor3}{HTML}{000000} 
\usepackage{tikz}

\algnewcommand\REQUIRED{\item[\textbf{Required:}]}%
\algnewcommand\INPUT{\item[\textbf{Input:}]}%
\algnewcommand\OUTPUT{\item[\textbf{Output:}]}%

\def\bg{{\bf g}}

\def\bc{{\bf c}}

\def\bo{{\bf o}}

\def\bm{{\bf m}}

\def\bs{{\bf s}}

\def\bx{{\bf x}}

\def\rs{{\mathsf{s}}}

\def\ro{{\mathsf{o}}}

\def\rx{{\mathsf{x}}}
\def\rc{{\mathsf{c}}}
\def\rm{{\mathsf{m}}}
\def\rtr{{\mathsf{tr}}}

\newtheorem{lemma}{Lemma}
\newtheorem{definition}{Definition}

 \title{\LARGE \textcolor{Mycolor3}{ Task-Effective Compression of Observations for the Centralized Control of a Multi-agent System Over Bit-Budgeted Channels} \\ 
\thanks{The authors are with the Centre for Security Reliability and Trust, University of Luxembourg, Luxembourg. Emails: \{arsham.mostaani, thang.vu, symeon.chatzinotas, bjorn.ottersten\}@uni.lu}
\thanks{This work is supported by European Research Council (ERC) via the project AGNOSTIC (Grant agreement ID: 742648).}
}
 \author{\IEEEauthorblockN{Arsham Mostaani, \IEEEmembership{Student Member,~IEEE},
 Thang X. Vu, \IEEEmembership{Senior Member,~IEEE}, \\
 Symeon Chatzinotas, \IEEEmembership{Fellow Member,~IEEE}, and Bj\"orn Ottersten, \IEEEmembership{Fellow Member,~IEEE} }
}
\begin{document}
\vspace{-3mm}
\maketitle

\vspace{-00mm}
\begin{abstract}
We consider a task-effective quantization problem that arises when multiple agents are controlled via a centralized controller (CC). While agents have to communicate their observations to the CC for decision-making, the bit-budgeted communications of agent-CC links may limit the task-effectiveness of the system which is measured by the system's average sum of stage costs/rewards. As a result, each agent should compress/quantize its observation such that the average sum of stage costs/rewards of the control task is minimally impacted. We address the problem of maximizing the average sum of stage rewards by proposing two different Action-Based State Aggregation (ABSA) algorithms that carry out the indirect and joint design of control and communication policies in the multi-agent system. While the applicability of ABSA-1 is limited to single-agent systems, it provides an analytical framework that acts as a stepping stone to the design of ABSA-2. ABSA-2 carries out the joint design of control and communication for a multi-agent system. 
We evaluate the algorithms - with average return as the performance metric - using numerical experiments performed to solve a multi-agent geometric consensus problem. The numerical results are concluded by introducing a new metric that measures the effectiveness of communications in a multi-agent system.

\end{abstract}


\vspace{-5mm}
\section{Introduction}
\vspace{-2mm}

As 5G is rolling out, a wave of new applications such as the internet of things (IoT), industrial internet of things (IIoT) and autonomous vehicles is emerging. It is projected that by 2030, approximately 30 billion IoT devices will be connected \cite{vailshery_2022}. With the proliferation of non-human types of connected devices, the focus of the communications design is shifting from traditional performance metrics, e.g., bit error rate and latency of communications to the semantic and task-oriented performance metrics such as meaning/semantic error rate \cite{Guler2018semantic, tong2021federatedsemantic} and the timeliness of information \cite{papas2021goal}. To evaluate how efficiently the network resources are being utilized, one could traditionally measure the sum rate of a network whereas in the era of the cyber-physical systems, given the resource constraints of the network, we want to understand how effectively one can conduct a (number of) task(s) in the desired way \cite{CALVANESE2021Semantic, mostaani2021task}. We are witnessing a paradigm shift in communication systems where the targeted performance metrics of the traditional systems are no longer valid. This imposes new grand challenges in designing the communications towards the eventual task-effectiveness  \cite{mostaani2021task}. This line of research is also driven partly due to the success of new machine learning technologies/ algorithms under the title of "emergent communications" in multi-agent systems \cite{FoersterLearning}. Transfer of these new technologies/ideas to communication engineering is anticipated to have a disruptive effect in multiple domains of the design of communication systems.

According to Shannon and Weaver, communication problems can be divided into three levels \cite{shannon1959mathematical}: (i) technical problem: given channel and network constraints, how accurately can the communication symbols/bits be transmitted? (ii) semantic problem: given channel and network constraints, how accurately the communication symbols can deliver the desired meaning? (iii) effectiveness problem: given channel and network constraints, how accurately the communication symbols can help to fulfil the desired task? While the traditional communication design addresses the technical problem, recently, the semantic problem \cite{Guler2018semantic, 8944272, 9804827, tong2021federatedsemantic,CALVANESE2021Semantic} as well as the effectiveness problem \cite{mostaani2021task,tung2021effective,mota2021emergence,shlezinger2021deep,gutierrez2022learning,zhang2022goal,shlezinger2020task,mostaani2019Learning,mostaani2022task} have attracted extensive research interest.

 \begin{figure}[t!] 
  \centering 
      \includegraphics[width=0.999 \textwidth]{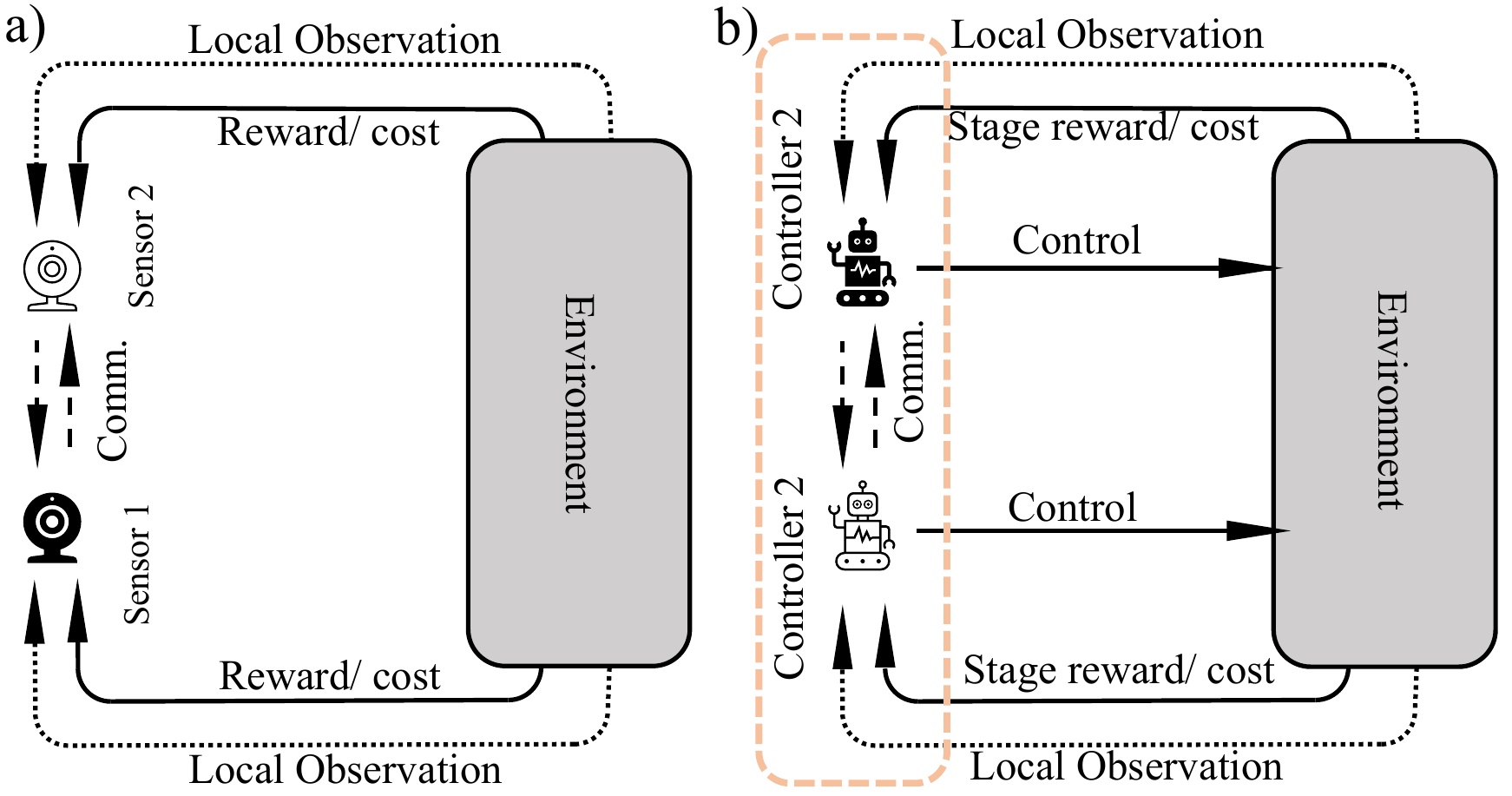} 
      \vspace{-2mm}
  \caption{Task-effective communications for a) an estimation vs. b) a control task - the orange dashed box is detailed in Fig. \ref{fig: communication topology} and Fig. \ref{fig: System model}.}
  \label{fig: control vs estimation}
  \vspace{-0.05cm}
\end{figure}

In contrast to Shannon’s technical-level communication framework, semantic communication can enhance performance by exploiting prior knowledge between source and destination \cite{Marios2021Semantic,papas2021goal}. The semantic-based designs, however, are not necessarily task-effective \cite{carnap1952outline}. One can design transmitters which compress the data with the least possible compromise on the semantic meaning being transmitted \cite{Guler2018semantic,tong2021federatedsemantic} while the transmission can be task-unaware \cite{zhang2022deep}. In contrast to semantic level and technical level communication design, the performance of a task-effective communication system is ultimately measured in terms of the average return/cost linked to the task \cite{tung2021effective}. In the (task-)effectiveness problem, we are not concerned only about the communication of meaning but also about how the message exchange is helping the receiving end to improve its performance in the expected cost/reward of an estimation task \cite{shlezinger2021deep,shlezinger2020task, stavrou2022rate,papas2021goal,gutierrez2022learning} or a control task \cite{mostaani2019Learning,tung2021effective,mota2021emergence,mostaani2020state,gutierrez2022learning,mostaani2022task,kim2019schedule}.

There are fundamental differences between the design of task-effective communications for an estimation vs. a control task - Fig. \ref{fig: control vs estimation}. (i) In the latter, each agent can produce a control signal that directly affects the next observations of the agent. Thus, in control tasks the source of information - local observations of the agent - is often a stochastic process with memory - e.g. linear or Markov decision processes - \cite{mostaani2019Learning,tung2021effective,mostaani2022task}. In the estimation tasks, however, the source of information is often assumed to be an i.i.d. stochastic process \cite{shlezinger2021deep,shlezinger2020task, stavrou2022rate}. (ii) In the control tasks, a control signal often has a long-lasting effect on the state of the system more than for a single stage/time step e.g., a control action can result in lower expected rewards in the short run but higher expected rewards in the long run. This makes the control tasks intrinsically sensitive to the time horizon for which the control policies are designed. Estimation tasks, specifically when the observation process is i.i.d., can be solved in a single stage/ time step - since there is no influence from the solution of one stage/ time step to another i.e., each time step can be solved separately \cite{stavrou2022rate,liu2022indirect}. (iii) The cost function for estimation tasks is often in the form of a difference/distortion function while in the control tasks it can take on many other forms.

In this paper, we focus on the effectiveness problem for the control tasks. In particular, we investigate the distributed communication design of a multiagent system (MAS) with the ultimate goal of maximizing the expected summation of per-stage rewards also known as the expected return. Multiple agents select control actions and communicate in the MAS to accomplish a collaborative task with the help of a central controller (CC) - i.e. the communication network topology of the MAS is a star topology with the hub node being the central controller and the peripheral nodes being the agents - Fig. \ref{fig: communication topology}. The considered system architecture can find applications in several domains such as Internet of Things, emerging cyber-physical systems, real-time interactive systems, vehicle-to-infrastructure communication \cite{chou2009feasibility} and collaborative perception \cite{liu2020who2com}. 

\subsection{Related works: Task-effective communications for control tasks}
Authors in \cite{mostaani2019Learning,mostaani2020state,mostaani2022task,mota2021emergence,tung2021effective,kim2019schedule,gutierrez2022learning} consider task-effective communication design under different settings. While \cite{mota2021emergence}, utilizes the task-effective communication design for the specific problem of the design of application-tailored protocols over perfect communication channels, the communication channel is considered to be imperfect in \cite{mostaani2019Learning,mostaani2020state,mostaani2022task,tung2021effective,kim2019schedule,gutierrez2022learning}. Authors in \cite{gutierrez2022learning} provide algorithmic contributions to the design of task-effective joint source channel coding for single agent systems. Task-effective joint source and channel coding for MAS is targeted by \cite{mostaani2019Learning, tung2021effective,gutierrez2022learning}, whereas \cite{mostaani2020state,mostaani2022task} are focused on task-effective data compression and quantization. Similar to the current paper, a star topology for the inter-agent communication is considered in \cite{mota2021emergence, tung2021effective} whereas \cite{mota2021emergence} assumes perfect communications between the hub node and the peripherals and \cite{tung2021effective} assumes imperfect communication channels at the down-link of the peripheral nodes. In contrast to all the above-mentioned work, this paper is - to the best of our knowledge - the first to study the star topology with the uplink (agent to hub) channel be imperfect (bit-budgeted) - Fig. \ref{fig: communication topology}. Accordingly, each agent observes the environment and communicates an abstract version of its local observation to the CC via imperfect (bit-budgeted) communication channels - red links in Fig. \ref{fig: communication topology}. Subsequently, CC produces control actions that are communicated to the agents via perfect communication channels - black links in Fig. \ref{fig: communication topology}. The control actions are selected by the CC such that they maximize the average return of the collaborative task, where the return is a performance metric linked to the accomplishment of the task. 

 \begin{figure}[t!] 
  \centering 
      \includegraphics[width=0.999 \textwidth]{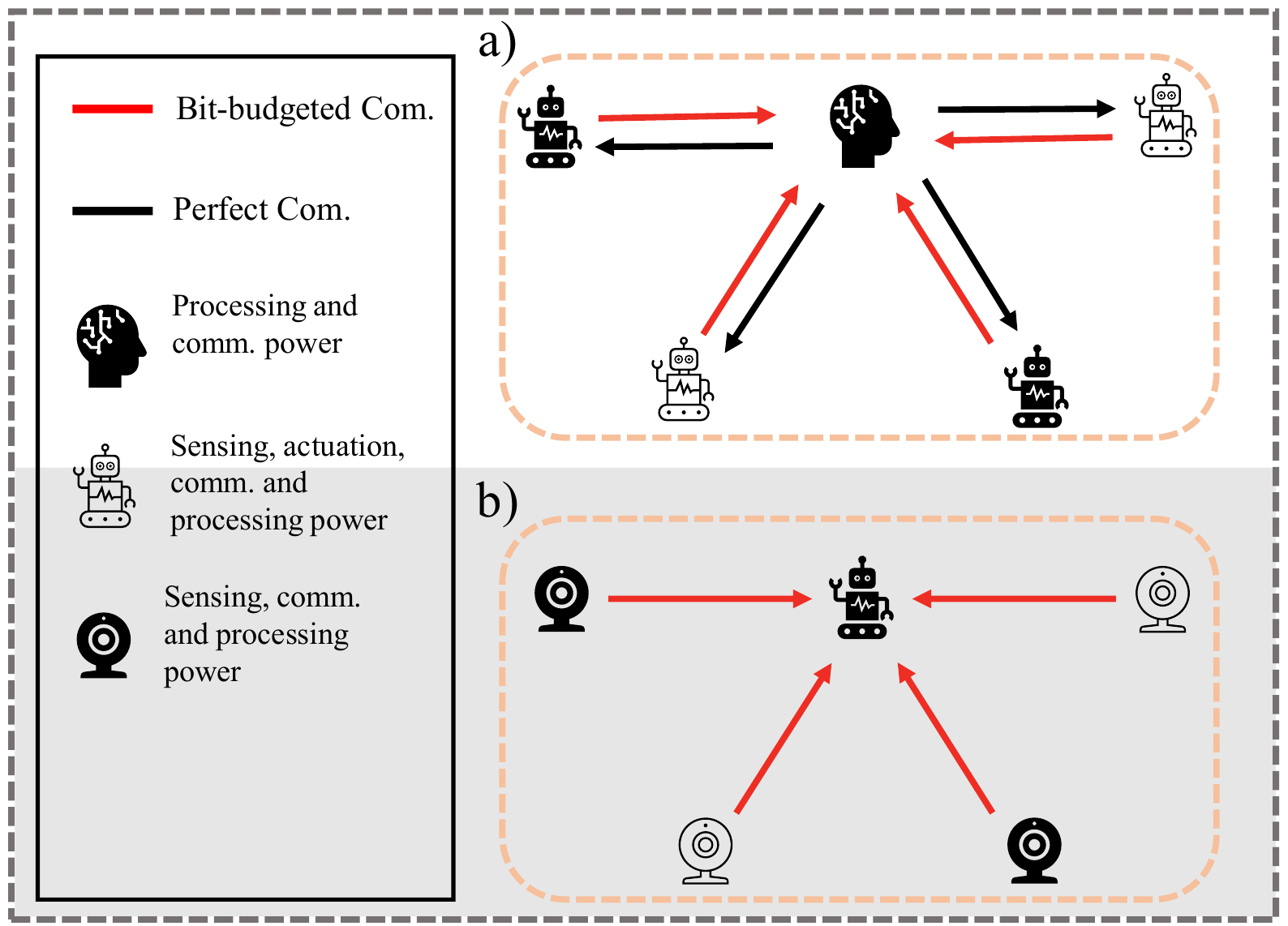} 
      \vspace{-2mm}
  \caption{Communication topology and its applicable scenarios a) Centralized control of an MAS with collocated actuators and sensors, b) Distributed sensing with a single controller collocated with a single actuator. The orange dashed box is detailing the same box in Fig. \ref{fig: control vs estimation} and Fig. \ref{fig: System model} .}
  \label{fig: communication topology}
  \vspace{-0.05cm}
\end{figure}

\subsection{Contributions}

In our earlier work \cite{mostaani2022task}, we have developed a generic framework to solve task-oriented communication problems - for a multi-agent system (MAS) with full mesh connectivity. The current work can be considered as an adoption of that framework to a new problem setting for the design of task-effective communications where agents follow a star network topology for their connectivity. In this direction, the current work transcends the applicability of the proposed framework beyond the specific problem that was solved in \cite{mostaani2022task} and provides further insights into how the framework can be used in wider terms and under a wider range of settings. In particular the contributions of this work are listed below.
\begin{itemize}
\item Firstly, we consider a novel problem setting in which an MAS is controlled via a central controller who has access to agents' local observations only through bit-budgeted distributed communications. This problem setting can be used in collaboration perception systems as well as vehicle-to-infrastructure communications, which cannot been addressed by the problem settings investigated in the prior similar art.

\item Secondly, our analytical studies establish the relationship between the considered joint communication and control design problem and conventional data quantization problems. In particular, lemma \ref{lemma: quantization measure} shows how the problem approached in this paper is a generalized version of the conventional data quantization. This formulation is useful as it helps to find an exact solution to the problem under stronger conditions via ABSA-1 and under milder conditions via ABSA-2.


\item Moreover, our analytical studies help us to craft an indirect \footnote{By an indirect algorithm here we mean an approach that is not dependent on our knowledge from a particular task. Indirect approaches are applicable to any/(wide range of) tasks. In contrast to indirect schemes, we have direct schemes that are specifically designed for a niche application  \cite{shlezinger2020task}. As defined by \cite{mostaani2021task}: "the direct schemes aim at guaranteeing or improving the performance of the cyber-physical system at a particular task by designing a task-tailored communication strategy".} task-effective data quantization algorithm - ABSA-2. Designing a task-effective data quantization for ABSA-2 can equivalently be translated as an indirect approach to feature selection for an arbitrary deep Q-network. Relying on the analysis carried out for ABSA-1, ABSA-2 designs distributed and bit-budgeted communications between the agents and CC. ABSA-2 is seen to approach optimal performance by increasing the memory of the CC. In fact, increasing the memory of CC leads to higher computational complexity. Therefore, ABSA-2 is said to strike a trade-off between computational complexity and task efficiency.

\item Numerical experiments are carried out on a geometric consensus task to evaluate the performance of the proposed schemes in terms of the optimality of the MAS's expected return in the task. ABSA-1 and ABSA-2 are compared with several other benchmark schemes introduced by \cite{mostaani2022task}, in a multi-agent\footnote{ Due to the complexity related issues explained in section \ref{Numerical results - Section}, the numerical results are limited to two-agent and three-agent scenarios.} scenario with local observability and bit-budgeted communications.

\item Finally, we will introduce a new metric, called task relevant information, for the measurement of effectiveness in task-oriented communication policies that - in comparison with the existing metrics such as positive listening and positive signalling - better explains the behaviour of a variety of task-effective communication schemes. The proposed metric is capable of measuring the effectiveness of a task-oriented communication/compression policy without the need of testing a jointly designed control policy and testing the jointly designed policies in the desired task.
\end{itemize}

\subsection{Technical approach}
 Our goal is to perform an efficient representation of the agents' local observations to ensure meeting the bit-budget of the communication links while minimizing the effect of quantization on the average return of the task. 
 \textcolor{Mycolor3}{ To achieve this, we first need to design task-effective data quantization policies for all agents. In task-effective data quantization, one needs to take into account the properties of the average return function and the optimal control policies associated with
the task \cite{zhang2022goal}. In addition to the design of the quantization policies for all agents, we also need the control policy of the CC to be capable of carrying out near-optimal decision-making despite its mere access to the quantized messages - resulting in a joint control and data compression problem. We formulate the joint control and data compression problem as a generalized form of data compression: task-oriented data compression (TODC).} Following this novel problem formulation, we propose two indirect action-based state aggregation algorithms (ABSA): (i) ABSA-1 provides analytical proof for a task-effective quantization i.e, with optimal performance in terms of the expected return. In this direction, ABSA-1 relaxes the assumption of the lumpability of the underlying MDP, according to which \cite{mostaani2022task}[condition. 6], the performance guarantees of the proposed method were established. Since ABSA-1 is only applicable when the system is composed of one agent and the CC we also propose ABSA-2. Following the analytical results of ABSA-1, given the help of MAP estimation to relax the aforementioned limitation of ABSA-1, and benefiting from a DQN controller at the CC; ABSA-2 will be introduced as a more general approach. (ii) ABSA-2 solves an approximated version of the TODC problem and carries out the quantization for any number of agents communicating with the CC. Thanks to a deep Q-network controller utilized at the CC, ABSA-2 can solve more complex problems where the controller benefits from a larger memory. Thus, ABSA-2 allows trading complexity for communication efficiency and vice versa. Finally, we will evaluate the performance of the proposed schemes in the specific task: a geometric consensus problem under finite observability \cite{barel2017come}.
 
 \subsection{Organization}
The rest of this paper is organized as follows. Section II describes the MAS and states the joint control and communication problem. Section III proposes two action-based state aggregation algorithms. Section IV shows the performance of the proposed algorithms in a geometric consensus problem. Finally, Section V concludes the paper. 
For the reader's convenience, a summary of the notation that we follow in this paper is given in Table \ref{table-notation}. Bold font is used for matrices or scalars which are random and their realizations follow simple font.

\begin{table}[b]
\caption{Table of notations}
\centering
 \begin{tabular}{||c c ||} 
 \hline
 Symbol & Meaning \\ [0.5ex] 
 \hline\hline
 \small{$\bx(t)$} & \small{A generic random variable generated at time $t$}  \\ 
 \hline
 \small{$\rx(t)$} & \small{Realization of $\bx(t)$}  \\
 \hline
 \small{$\mathcal{X}$} & \small{Alphabet of \bx(t)}  \\
 \hline
 \small{$|\mathcal{X}|$} & \small{Cardinality of $\mathcal{X}$}  \\
 \hline
 \small{$p_{\bx}\big(\rx(t)\big)$} & \small{Shorthand for $\mathrm{Pr}\big(\bx(t) = \rx(t) \big)$}  \\  
 \hline
 \small{$H\big(\bx(t)\big)$} & \small{Information entropy of $\bx(t) $ (bits)}  \\  
 \hline
  \small{$\mathcal{X}_{-\bx}$} & \small{ $\mathcal{X} - \{\bx\}$}  \\ [1ex]
 \hline
   \small{$\mathbb{E}_{p(\rx)}\{\bx\}$} & \small{\makecell{Expectation of the random variable $X$ over the \\ probability distribution $p(\rx)$}}  \\ [1ex]
 \hline
   \small{$\rtr(t) $} & \small{Realization of the system's trajectory at time $t$}  \\ [1ex]
 \hline
\end{tabular}
\label{table-notation}
\end{table}

\vspace{-5mm}
 \section{System model and problem statement} \label{System model - Section}
  \vspace{-3mm}
The problem setting we introduce here can be used to analyse both scenarios illustrated in Fig. \ref{fig: communication topology}. Nevertheless, to use our language consistently, we focus on scenario (a) of that figure throughout the manuscript. In particular, when we use the term "agent" we refer to an object which certainly has all the following hardware capabilities: sensing, actuation, communication and data processing. A MAS, however, may not be comprised of mere agents, but of a combination of agents and perhaps other objects that has at least the hardware capabilities for communication and data processing power. The central controller here is supposed to have the hardware capability to process relatively larger data as well as the capability of communications. The interactions inside the MAS and outside the MAS with the environment are illustrated in Fig. \ref{fig: System model}.

\subsection{System model}\label{subsect: system model}
We consider a MAS in which  multiple agents $i \in \mathcal{N} = \{1,2,..., N\}$ collaboratively solve a task with the aid of a CC. \textcolor{Mycolor3}{ Following a centralized action policy, CC provides the agents with their actions via a perfect communication channel while it receives the observations of agents through an imperfect communication channel \footnote{In this work we follow a common assumption used in the networked control literature \cite{tatikonda2004control} according to which the bit-budget only limits the uplink communications of the agents and not their downlink. Accordingly, the agents select their control actions as is dictated to them by the central controller.}. The considered setting is similar to conventional centralized control of MASs \cite{mostaani2022task, FoersterCounter}, except for the fact that the communications from the agents to the CC are transmitted over a bit-budgeted communication channel.} The agent-hub communications are considered to be instantaneous and synchronous \cite{mostaani2022task}. This is in contrast with the delayed \cite{mostaani2019Learning,oliehoek2016concise} and sequential/iterative communication models \cite{ding2021sequential,albowicz2001recursive,dorvash2013stochastic}. We note that there is no direct inter-agent communication in the considered system - communications occur only between agents and the central controller. The system runs on discrete time steps $t$. \textcolor{Mycolor3}{ The observation of each agent $i$ at time step $t$ is shown by $\bo_i(t)\in \Omega$ and the state $\bs(t) \in \mathcal{S}$ of the system is defined by the joint observations $\bs(t) \triangleq \langle \bo_1(t),\dots, \bo_N(t)   \rangle $\footnote{According to this definition, at any given time $t$ the observations of any two agent $i,j \in \mathcal{N}$ are linearly independent in the Euclidean space. The same conditions are true for the control actions of arbitrary agents.} .} The control action of each agent $i$ at time $t$ is shown by $\bm_i (t) \in \mathcal{M}$, and the action vector $\bm(t) \in \mathcal{M}^N $ of the system is defined by the joint actions $\bm(t) \triangleq \langle \bm_1(t), ..., \bm_N(t)   \rangle $. The observation space $\Omega$, state-space $\mathcal{S}$, and action space $\mathcal{M}$ are all discrete sets.
 \begin{figure}[t] 
  \centering 
      \includegraphics[width=0.95\textwidth]{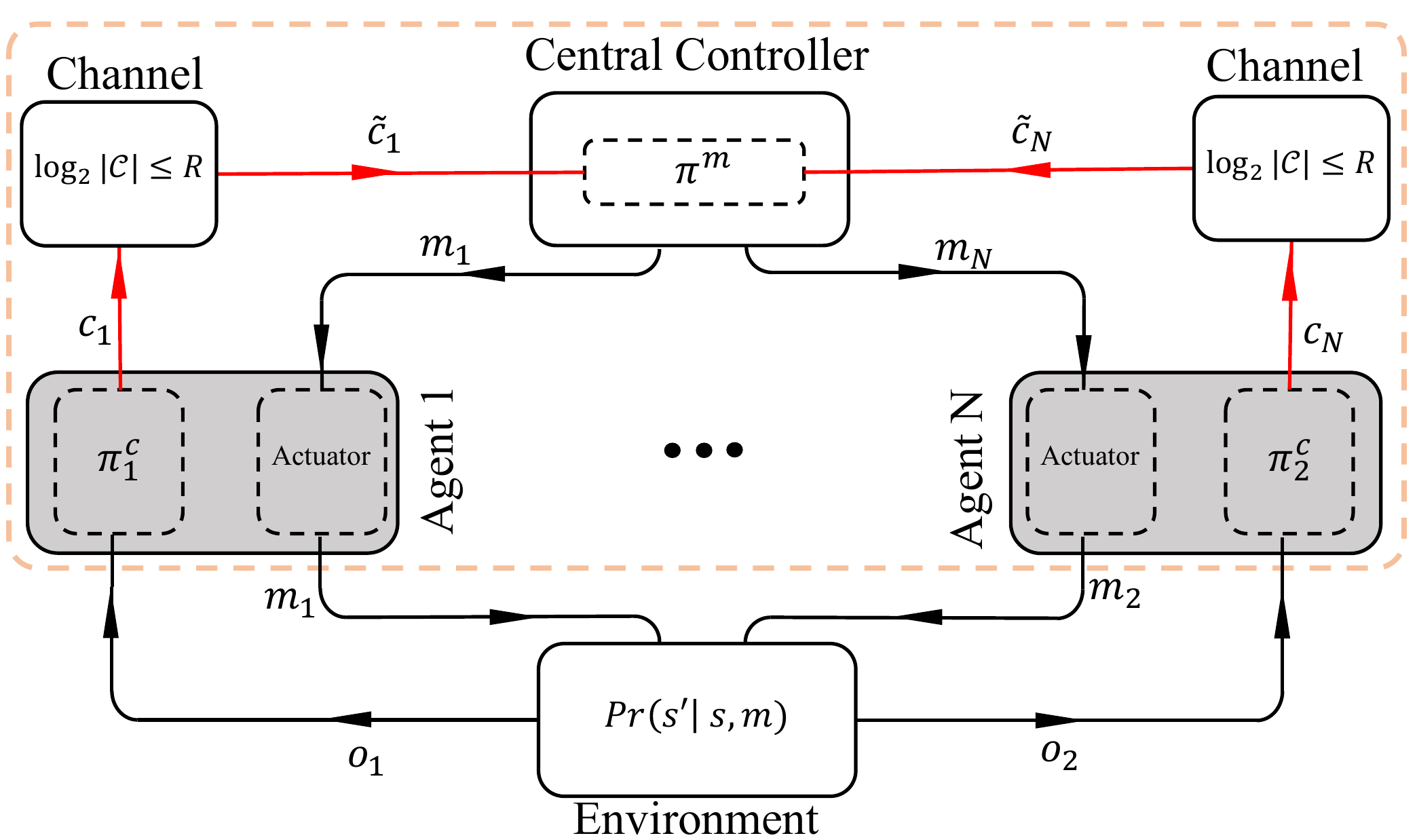} 
      \vspace{-2mm}
  \caption{Illustration of the interactions of the CC and agents for the control of the environment. The red link shows the communication channels that are bit-budgeted - implying the local (and not global) observability of the CC. The orange dashed box is detailing the same box in Fig. \ref{fig: control vs estimation} and Fig. \ref{fig: communication topology} .}
  \label{fig: System model}
  \vspace{-0.05cm}
\end{figure}
\textcolor{Mycolor3}{ The environment is governed by an underlying\footnote{\textcolor{Mycolor3}{ As defined in the literature [10], the underlying MDP’ is the horizon-$T'$ MDP defined by a hypothetical single agent that takes joint actions
$\bm(t) \in \mathcal{M}^N $ and observes the nominal state $\bs(t) \triangleq \langle \bo_1(t),\dots, \bo_N(t)   \rangle $ that has the same transition model $T(\cdot)$ and reward
model $r(\cdot)$ as the environment experienced by our MAS.}} Markov Decision Process that is described by the tuple $M = \big{\{} \mathcal{S},\mathcal{M}^N, r(\cdot), \gamma, T(\cdot)  \big{\}}$, where $r(\cdot): \mathcal{S} \times \mathcal{M}^N \rightarrow \mathbb{R}$ is the per-stage reward function and the scalar $ 0\leq\gamma \leq 1 $ is the discount factor.} The function $T(\cdot): \mathcal{S}\times \mathcal{M}^N \times \mathcal{S} \rightarrow [0,1]$ is a conditional probability mass function (pmf) which represents state transitions such that \textcolor{Mycolor3}{ $T\big(\rs(t+1), \rs(t) , \rm(t) \big) = \mathrm{Pr}\big( \rs(t+1) | \rs(t) , \rm(t) \big)$}. According to the per-stage reward signals, the system's return within the time horizon $T'$ is denoted by
{\small\begin{equation} \label{eq: cumulative rewards}
    \bg(t^{'})= {\sum}_{t=t^{'}}^{T'}\gamma^{t-1} r\big(\bo_1(t), ...,\bo_N(t),\bm_1(t), ...,\bm_N(t)\big).
\end{equation}}
While the system state is jointly observable by the agents \cite{pynadath2002communicative}, each agent $i$'s observation $\bo_i(t)$ is local \footnote{In our problem setting, each agent does not see the environment as an MDP due to their local observability. We only assume the presence of an underlying MDP for the environment, which is widely adopted in the literature for the reinforcement learning algorithm, e.g., \cite{oliehoek2007dec} \cite{lowe2019pitfalls}.
We have this assumption as our performance guarantees rely on the optimality of the solution provided for the control task, which is also assumed in [7], [10]. Let us recall that throughout all of our numerical studies, even the CC, given joint observations of all agents, cannot observe the true/nominal state of the environment.}. Once per time step, agent $i \in \mathcal{N}$ is allowed to transmit its local observations through a communication message $\bc_i(t)$ to the CC. The communications between agents and the central controller are done in a synchronous (not sequential) and simultaneous (not delayed) fashion \cite{mostaani2019Learning}. Each agent $i$ generates its communication message $\bc_i(t)$ by following its communication policy $\pi^c_i(\cdot) : \Omega \rightarrow \mathcal{C}$. In parallel to all other agents, agent $i$ follows the communication policy $\pi^c_i(\cdot)$ to map its current observation $\bo_i(t)$ to the communication message $\bc_i(t)$ which will be received by the central controller in the same time-step $t$. The code-book $\mathcal{C}$ is a set composed of a finite number of communication code-words s $\rc, \rc', \rc'', ..., \rc^{( |\mathcal{C}| -1)}$ - we use the same notation to refer to the different members of the action, observation and state spaces too. Agents' communication messages are sent over an error-free finite-rate bit pipe, with its rate constraint to be $R \in \mathbb{R}$ (bits per channel use) or equivalently (bits per time step). As a result, the size of the quantization codebook should follow the inequality $ |\mathcal{C}|\leq 2^R $. The CC exploits the received communication messages $\bc(t) \triangleq \langle {\bc}_1(t),...,  {\bc}_N(t) \rangle$  within the last $d$ number of time-steps to generate the action signal $\rm(t)$ following the control policy $\pi^m(\cdot): \mathcal{C}^{Nd} \rightarrow \mathcal{M}^N$. \textcolor{Mycolor3}{ Based on the above description, the environment from the point of view of the CC as well as from the agent's point of view is not necessarily an MDP - as none is capable of viewing the nominal state of the environment.}

\vspace{-0mm}
\subsection{Problem statement: Joint Control and Communication Design (JCCD) problem}
\textcolor{Mycolor3}{ Now we define the JCCD problem.} Let $M$ be the MDP governing the environment and the scalar $R \in \mathbb{R}$ to be the bit-budget of the uplink of all agents. At any time step $t'$, we aim at selecting the tuple $\pi = \langle \pi^m(\cdot), \pi^c\rangle $ with $ \pi^c \triangleq \langle \pi^c_1(\cdot),..., \pi^c_N(\cdot) \rangle$ to solve the following variational dynamic programming
\begin{align}\label{eq: Joint Control and Communication Design (JCCD) problem}
 \underset{\pi}{\textnormal{argmax}} ~~ 
\mathbb{E}_{\pi}                  \Big\{
   \bg(t') 
\Big\};~~
 \textnormal{s.t.} ~~  
    |\mathcal{C}|\leq 2^R,
\end{align}
where the expectation is taken over the joint pmf of the system's trajectory $\{\rtr\}_{t'}^{T'} = \ro_1(t'),..., \ro_N(t'), m(t'), ..., \ro_1(T'),..., \ro_N(T'), \rm(T')$, when the agents follow the policy tuple $\pi$. In the next section, similar to \cite{mostaani2022task} we will disentangle the design of action and communication policies via action-based quantization of observations. In contrast to \cite{mostaani2022task}, here the communication network of the MAS is assumed to follow a star topology. The idea behind this disentanglement is to extract the features of the control design problem that can affect the communication design and to take them into account while designing the communications. Thus our communication design will be aware of the key features of the control task. We extract the key features of the control task using analytical techniques as well as reinforcement learning \cite{mostaani2022task, mostaani2019Learning}. In fact, the new communication problem called TODC, will no longer be similar to the conventional communication problems, as it is inspired by the JCCD problem.

In \cite{mostaani2020state, mostaani2022task}, authors use the value of agents' observations for the given task as the key feature of the control task considered in the communication design. Accordingly, the idea was to cluster together the observation points that have similar values. In contrast to \cite{mostaani2020state, mostaani2022task}, which considers the value of observations as an explicit key feature of the control task, here we consider the optimal control/action values assigned to each observation as the key feature. Accordingly, ABSA clusters the observation values together, whenever the observation points have similar optimal control/action values assigned to them. Action-based state aggregation has been already introduced in the literature of reinforcement learning as a means for reducing the complexity of the reinforcement learning algorithms while maintaining the average return performance \cite{li2006towards,mccallum1996reinforcement}.  



\vspace{-3mm}
\section{Action-based Lossless compression of observations}
\vspace{-2mm}

In this section, we will set yet another example - in addition to \cite{mostaani2022task} - for the use of a generic framework to solve JCCD problem. In \cite{mostaani2022task}, a similar problem is solved for distributed control and quantization, wherein, the authors disentangle the design of task-oriented communication policies and action policies given the aid of a hypothetical functional $\Pi^{m^*}$. In particular, the functional $\Pi^{m^*}$ is a map from the vector space $\mathcal{K}^c$ of all possible communication policies $ \pi^c$ to the vector space $\mathcal{K}^m$ of optimal corresponding control policy $\pi^{m^*}(\cdot)$. Upon the availability of the functional $\Pi^{m^*}$, wherever the function $\pi^m$ appears in the JCCD problem, it can be replaced with $\Pi^{m^*}(\pi^c)$ resulting in a novel problem in which only the communication policies $\pi^c$ are to be designed.  While in \cite{mostaani2022task}, authors use an approximation of $\Pi^{m^*}(\pi^c)$ to obtain a task-oriented quantizer design problem, in the current work we derive an exact solution for a simplified version of (\ref{eq: TBIC problem - adam's law applied - 4 - body}) - where the number of agents communicating with the central controller is limited to one agent. To adapt ABSA to the generic setting of the problem (\ref{eq: TBIC problem - adam's law applied - 4 - body}), in ABSA-2, we will lift this limitation given the aid of an approximation technique.

The JCCD problem can already be formulated as a form of data-quantization problem. Lemma \ref{lemma: quantization measure}, identifies the quantization metric that we aim to optimize in this paper. It  reformulates the JCCD problem as a novel generalized data quantization problem.

\begin{lemma} \label{lemma: quantization measure}
 The JCCD problem (\ref{eq: Joint Control and Communication Design (JCCD) problem}) can also be expressed as a generalized data quantization problem as follows
 {\small
\begin{align}\label{eq: TBIC problem - adam's law applied - 4 - body}
& \underset{\pi}{\text{argmin }} 
& & \mathbb{E}_{p(\rs(t))}                   \Big{|}
    V^{\pi^*}\big(\bs(t) \big) 
-                
    V^{ {\pi}^m}\big(\bc(t) \big)  
\Big{|} , ~~ \text{s.t.} ~~ |\mathcal{C}|\leq 2^R,
\end{align}}
where the communication vector $\bc(t)$ generated by $\pi^c$ is a quantized version of the system's state $\bs(t)$.
\end{lemma}

\begin{proof}
Appendix \ref{append: proof of quantization measure}.
\end{proof}
In contrast to the classic data-quantization problems, here the distortion metric, measures the difference between two different functions of the original signal and its quantized version - namely $V^{\pi^*(\cdot)}$ and $V^{\pi^m(\cdot)}$ - thus the distortion measure that we aim to optimize by solving (\ref{eq: TBIC problem - adam's law applied - 4 - body}) is not conventional. In fact, the variational minimization problem is solved over the vector space of joint quantization policies $\pi^c$ and action policy $\pi^m$ functions.

\vspace{-4mm}
\subsection{ABSA-1 Algorithm} \label{subsect: ABSA-1}
\vspace{-2mm}
The applicability of the proposed ABSA-1, is limited to two mathematically equivalent scenarios: (i) we have a single agent communicating to the CC - consider the Fig. \ref{fig: communication topology}-a, with only one agent connected to the CC - or (ii) that the agents communicate with the CC through a relay. In the latter scenario, the relay has full access to the agents' communication observation, i.e., $\bo_i, \forall i\in \mathcal{N}$, while the relay to CC channel is bit-budgeted.  
This limited scenario is useful for us to facilitate our analytical studies on the problem (\ref{eq: TBIC problem - adam's law applied - 4 - body}), allowing us to establish theoretical proof for the losslessness of compression in ABSA-1 as well as its optimal average return performance. These statements will be confirmed by Lemma ~\ref{Theorem: ABSA-1 is Lossless} - the results of which will also be useful to design ABSA-2.
The central idea of ABSA-1 is to represent any two states $\bs^{(i)},\bs^{(j)}$ using the same communication message $\bc$ iff $\pi^*\big(\bs^{(i)}\big)=\pi^*\big(\bs^{(j)}\big)$, where $\pi^*(\cdot): \mathcal{S} \rightarrow \mathcal{M}^N$ is the optimal control policy of the agents, given the access of observations from all agents. Thus, ABSA-1 and ABSA-2 solve the JCCD problem at three different phases: (i) solving the centralized control problem under perfect communications via reinforcement learning i.e., Q-learning, to find $\pi^*(\cdot)$\footnote{ABSA's bottleneck arises from the increasing complexity of Q-learning as agents increase in number $N$. Similar limitations are in place for any other algorithm that requires a centralized training phase \cite{FoersterCounter,FoersterLearning}}, (ii) solving the task-oriented data quantization problem to find $\pi^c$ via a form of data clustering, (iii) finding the $\pi^m$ corresponding to $\pi^c$.

\textcolor{Mycolor3}{ In order to explain ABSA-1, we introduce the problem of task-oriented data compression with centralized control. TBIC is derived using similar techniques in \cite{mostaani2022task} but for a different setting i.e., the communication network of MAS has a star topology. The TBIC problem is no longer a joint control and communication problem but is a quantization design problem in which the features of the control problem are taken into account.}
\textcolor{Mycolor3}{ To arrive to TODC problem from the JCCD problem,} we use the functional $\Pi^{m^*}$ to replace $\pi^m(\cdot)$ with $\Pi^{m^*} \big( \pi^c \big)$ . Upon the availability of $\Pi^{m^*}$, by plugging it into the JCCD problem (\ref{eq: Joint Control and Communication Design (JCCD) problem}), we will have a new problem 
 {\small
\begin{align}\label{eq: TBIC problem}
& \underset{\pi^c}{\text{argmin }} 
& & \mathbb{E}_{p(\rs(t))}                   \Big{|}
    V^{\pi^*}\big(\bs(t) \big) 
-                
    V^{ \Pi^{m^*} \big( \pi^c \big)}\big(\bc(t) \big)  
\Big{|} , ~~ \text{s.t.} ~~ |\mathcal{C}|\leq 2^R,
\end{align}}
where we maximize the system's return with respect to only the communication policies $\pi^c(\cdot)$ of the local relay. The optimal control policy $ \pi^{m^*}(\cdot)$ of the CC is automatically computed by the mapping $\Pi^{m^*}\big( \pi^c(\cdot)\big)$. The problem is called here as the TODC problem.
\textcolor{Mycolor3}{Upon the availability of $\Pi^{m^*}$, the JCCD problem \eqref{eq: Joint Control and Communication Design (JCCD) problem} can be reduced to (\ref{eq: TBIC problem}).  Definition \ref{def: ABSA-1} is provided to formalize a precise approach to solve (\ref{eq: TBIC problem}) via obtaining the communication policy of the relay $\pi^c(\cdot)$ as well as the corresponding $\Pi^{m^*}$, to solve (\ref{eq: Joint Control and Communication Design (JCCD) problem}).}

\begin{definition}\textbf{Quantization and control policies in ABSA-1:}\label{def: ABSA-1}

The communication policy $\pi^{c,ABSA-1}(\cdot)$ designed by ABSA-1 will be obtained by solving the following k-median clustering problem
\begin{align} \label{eq: ABSA-1}
    & \underset{\mathcal{P}}{\text{ min }}
    & {\sum}_{i=1}^{|\mathcal{C}|} {\sum}_{\bs(t) \in \mathcal{P}_i} 
    \Big{|}
    \pi^*\big( \bs(t) \big)  - \mu_i 
    \Big{|},
\end{align}
where $\mathcal{P} = \{\mathcal{P}_1, ..., \mathcal{P}_B \}$ is a partition of $\mathcal{S}$ and $\mu_i$ is the centroid of each cluster $i$. The communication policy of ABSA-1 - $\pi^{c,ABSA-1}(\cdot)$ - is an arbitrary non-injective mapping such that $\forall k \in \{1,...,B\}: \pi^{c,ABSA-1}(\bs) = \bc^{(k)}$ if and only if $\bs \in \mathcal{P}_k$.
 Now let $C_g$ be a function composition operator such that $C_g f = g \circ f$. We define the operator $ \Pi^{m^*} \triangleq C_g$, with $g = \pi^*\big( \pi^{{c,ABSA-1}^{-1}}(\cdot) \big)$\footnote{ Note that as $\pi^{c,ABSA-1}(\cdot)$ is non-injective, its inverse would not produce a unique output given any input. Thus, by $\pi^*\big( \pi^{{c,ABSA-1}^{-1}}(\bc') \big)$ we mean $\pi^*\big( \bs' \big)$, where $\bs'$ can be any arbitrary output of $\pi^{{c,ABSA-1}^{-1}}(\bc')$.} . 
\end{definition}

The optimality of the proposed ABSA-1 algorithm is subsequently provided in Theorem~\ref{Theorem: ABSA-1 is Lossless}.

\vspace{-3mm}
\textcolor{Mycolor3}{
\begin{lemma} \label{Theorem: ABSA-1 is Lossless}
 The communication policy $\pi^{c,ABSA-1}$ - as described by Definition \ref{def: ABSA-1} - will carry out lossless compression of observation data w.r.t. the average return if $|\mathcal{C}| \geq |\mathcal{M}|^N$.
\end{lemma}}

\begin{proof}
Appendix \ref{append: proof: Theorem: ABSA-1 is Lossless}.
\end{proof}
\textbf{Remark:} ABSA-1 will also carry out lossless compression of observation data with respect to the distortion measure introduced in problem (\ref{eq: TBIC problem - adam's law applied - 4 - body}). Given the proofs of lemma 2 and lemma 1, the proof of this remark is straightforward and is therefore, omitted.  

\textcolor{Mycolor3}{The losslessness of quantization in ABSA-1 implies that the $\pi^{ABSA-1}$ will result in no loss of the system's average return, compared with the case where the optimal policy $\pi^*(\cdot)$ is used to control the MAS under perfect communications. Consequently, the control policy $\pi^{m, ABSA-1}(\cdot)$ is optimal. Let us recall once again that here, we do not use a conventional quantization distortion metric, we select a representation of local observation in such a way that the conveyed message maximizes the average task return.}

\textcolor{Mycolor3}{Note that in [7], the authors do not find the higher order function $\Pi^{m^*}$ that reduces the joint communications and control problem to a task-oriented communication design - instead they solve an approximated version of the task-oriented communication design problem. In this paper, however, we introduce a closed form $\Pi^{m^*}$ by ABSA-1 that can map every communication policy $\pi^{c,ABSA-1}$ introduced by ABSA-1, to the exact optimal control policy. This implies that the solutions provided by ABSA-1 are also the optimal solutions of the joint communication and control design (JCCD) problem.}


\vspace{-4mm}
\subsection{ABSA-2 Algorithm}
\vspace{-2mm}
We saw earlier in lemma \ref{Theorem: ABSA-1 is Lossless} that the communication policy obtained by solving the problem \ref{eq: ABSA-1} is optimal and can result in a lossless average return performance when $|\mathcal{C}| \geq |\mathcal{M}|^N$. To solve the problem \ref{eq: ABSA-1}, however, we need to know $\pi^*\big(\bs(t)\big)$. This is a limiting assumption that in ABSA-1 can be translated to two different system models which are less general than the system pictured in Fig. \ref{fig: System model}: (i) presence of an extra relay between the agents and the central controller where the relay has perfect downlink channels to agents and a single bit-budgeted channel to the CC. (ii) The MAS is only composed of one single agent and a CC where the uplink of the agent is bit-budgeted but its downlink is a perfect channel.

Our second proposed algorithm ABSA-2 removes the need to know $\pi^*\big(\bs(t)\big)$ and can run under the more general settings shown in Fig. \ref{fig: System model}. This is done by approximating the local element $\bm^*_i(t)$ of $\pi^*\big(\bs(t)\big) = \langle \bm_1*(t), ..., \bm_N*(t) \rangle$ at agent agent $i$ given the local observation of this agent $\bo_i(t)$. That is, given a centralized training phase, we will have access to the empirical joint distribution of $p(\ro_i,\rm^*_i)$, using which we can obtain a numerical MAP estimator of $\hat{\rm^*}_i$. Thus ABSA-2 allows for fully distributed communication policies. \textcolor{Mycolor3}{In particular, the encoding of the communication messages of each agent is carried out separately by them before they communicate with CC or any other agent. This form of encoding is often referred to as distributed encoding. Furthermore, the encoding carried out by ABSA-2 at each agent is a low-complexity and low-power process that requires no inter-agent communications before hands. In this case, each agent directly communicates its encoded observations to the CC via a bit-budgeted communication channel.} \textcolor{Mycolor3}{ In order to improve the learning efficiency at CC, it can take into account all the communications received in the time frame $[t-d, t ]$  to make a control decision $m(t)$. Therefore, the ABSA-2 algorithm can strike a trade-off between the complexity of the computations carried out at the CC - directly impacted by the value of $d$ - and effectiveness of agents' communications - inversely impacted by the value of $|\mathcal{C}|$.} 
Moreover, ABSA-2 is straightforwardly extendable to the different values of $|\mathcal{C}|$ per each agent $i$, instead of having only one fixed bit-budget $R = \log_2 |\mathcal{C}|$ for all agents.

 \begin{figure*}[t] 
  \centering 
      \includegraphics[width=0.80\textwidth]{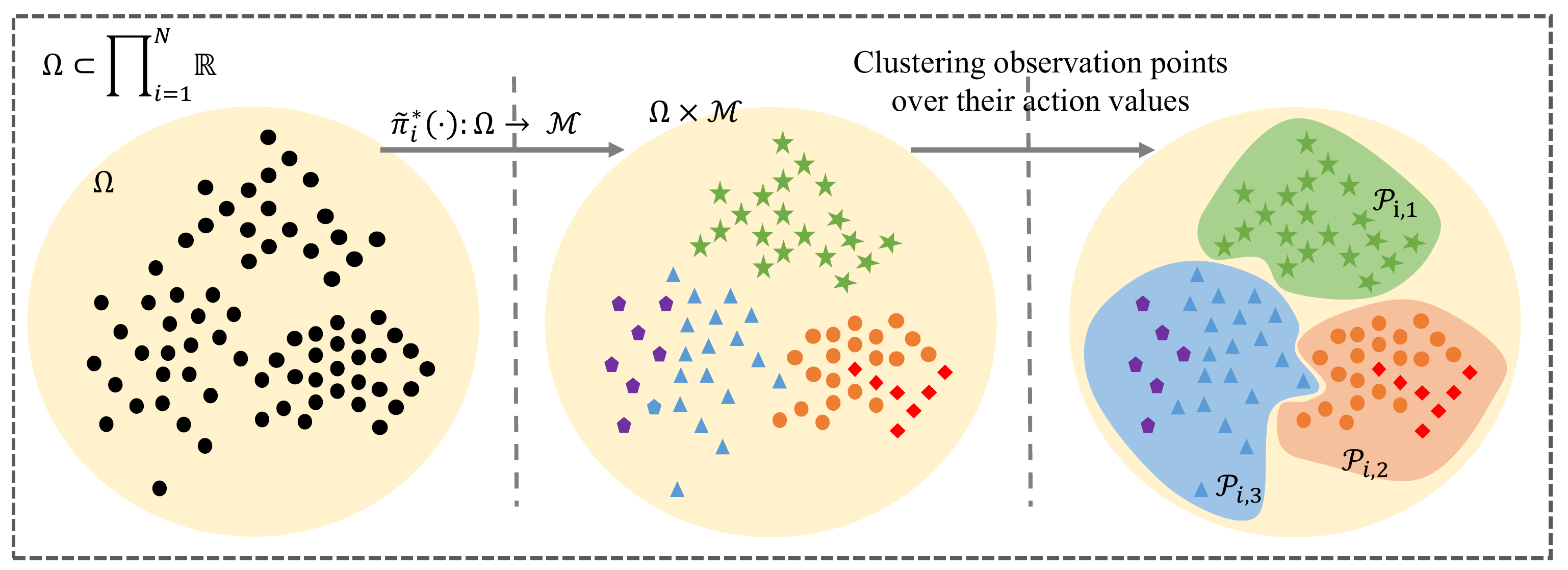} 
      \vspace{-1mm}
  \caption{Abstract representation of states in ABSA-2 with $ |\mathcal{C}| =3$ and $|\mathcal{M}| = 5$ - $|\mathcal{M}|$ is represented by the number of shapes selected to show the observation points and $|\mathcal{C}|$ is represented by the number of clusters shown in the right subplot. The left subplot shows the observation points prior to aggregation. During a centralized training phase we first compute $\pi^*(\cdot)$ according to which $ {\pi}^*_i(\cdot) : \Omega \rightarrow \mathcal{M}$ can be obtained. We use the surjection $ {\pi}^*_i(\cdot)$ to map a high dimensional/precision observation space to a low dimensional/precision space. The middle subplot shows the observation points together with the action values assigned to them - each unique shape represents a unique action value. \textbf{This new representation of the observation points, embeds the features of the control problem into the data quantization problem.} Finally, we carry out the clustering of observation points according to their action values - all observation points assigned to (a set of) action values are clustered together. The right subplot shows the aggregated observation space, where all the observation points in each cluster will be represented using the same communication message. The centralized controller which is run using DQN, observes the environment at each time step, through all these aggregated observations/communications it receives from all the agents.}
  \label{fig: Abstract representation by ABSA-2}
  \vspace{-0.05cm}
\end{figure*}


\begin{algorithm}\label{ABSA2 - Algorithm}
\caption{ Action Based State Aggregation (ABSA-2)}
\begin{algorithmic}[1]

 \State \textbf{Initialize} replay memory $D$ to capacity 10'000.
  \State \textbf{Initialize} state-action value function $Q(\cdot)$ with random weights $\theta$.
  \State \textbf{Initialize} target state-action value function $Q^{t}(\cdot)$ with weights $\theta^t = \theta $.

 \State Obtain $\pi^{*}(\cdot) \text{ and } Q^{*}(\cdot)$ by solving (2) using Q-learning \cite{Suttonintroduction}*, where $R >> H(\bo_i(t)) \,\,\, \forall i \in \mathcal{N}$. 
 \State Compute $ {\pi}^*_i (\ro_i(t)) = \text{Mode} \big[  \bm^*_i | \ro_i(t) \big]$, for $\forall \ro_i(t) \in \Omega$, for $i \in \mathcal{N}$.
 \State Solve problem (5) by applying k-median clustering to obtain $\mathcal{P}_i$ and $\pi^c_i(\cdot)$  , for $i \in \mathcal{N}$.
  \For{each episode $k=1:$ 200'000} 
 \State Randomly initialize observation $\ro_i(t=0)$, for $i \in \mathcal{N}$
  \State Randomly initialize the message $ {\rc}(t=0)$
\For{$t = 1: T'$}

\vspace{1mm}          
            \State Select $\rc_i(t)$, at agent $i$, following $\pi^{c}_i(\cdot)$, for $i \in \mathcal{N}$

            \State Obtain the message $\langle  {\rc}_1(t),...,  {\rc}_N(t) \rangle$ at the CC
            
            \State Follow $\epsilon$-greedy, at CC, to generate the action $\rm_i(t)$, for $i \in \mathcal{N}$  
            \State Obtain reward $r(t) = R\big( \rs(t),\rm(t) \big)$ at the CC
            \State Store the transition {\small $\Big{\{}  {\rc}(t), \rm(t), r(t),  {\rc}(t+1)  \Big{\}}$ in $D$}
\State $t \leftarrow t+1$

\EndFor\label{euclidendwhile-2}
\State \textbf{end}

\State Sample  \small {$D' = \Big{\{}  {\rc}(t'), \rm(t'), r(t'),  {\rc}(t'+1)  \Big{\}}_{t'=t'_1}^{t'=t'_{62}} $} from $D$ 

\For{each transition $t'=t'_1:t'_{62}$ of the mini-batch $D'$} 
\State {\small Compute DQN's average loss  $L_{t'}(\theta) = \frac{1}{2} \Big(r(t') \! + \! \underset{\rm^*}{\text{ max }} Q^t \big(  {\rc}(t'+1) , \rm^*, \theta^t \big) \! - \!
\underset{\rm^*}{\text{ max }} Q \big(  {\rc}(t') , \rm^*, \theta \big)  \Big)^2$}, 
\State {\small Perform a gradient descent step on $L_{t'}(\theta)$ w.r.t $\theta$}
\EndFor
\State \textbf{end}
\State Update the target network $Q^t(\cdot)$ every $1000$ steps
\EndFor
\State \textbf{end}
\vspace{1mm}

%
\vspace{-1mm}
\end{algorithmic}
\end{algorithm}


As illustrated in Fig. \ref{fig: Abstract representation by ABSA-2}, ABSA-2, each agent $i$ obtains a communication policy function $\pi^c_i(\cdot)$ by solving  a clustering problem  over its local observation space instead of the global state space, formulated as follows:
\begin{align} \label{eq: ABSA-2 - comms}
    & \underset{\mathcal{P}_i}{min}
    & {\sum}_{j=1}^{|\mathcal{C}|} {\sum}_{\ro_i(t) \in \mathcal{P}_{i,j}} 
    \Big{|}
     \tilde{\pi}^*_i (\ro_i(t)) - \mu_{i,j}  
    \Big{|}
    ,
\end{align}
where $\mathcal{P}_i = \{\mathcal{P}_{i,1}, ..., \mathcal{P}_{i,|\mathcal{C}|} \}$ is a partition of $\Omega$, and
\begin{align}
    \tilde{\pi}^*_i (\ro_i(t)) = {\text{argmax}}_{\bm_i^*} \,\, p_{\pi^*}(\bm_i^* | \ro_i(t)) ,
\end{align}
and $\bm^*_i $ is the optimal action of agent $i$, which is $i$-th element of $\bm^* \triangleq \pi^*\big( \bo_1(t), ..., \bo_N(t) \big)$. Thus $ \tilde{\pi}^*_i (\ro_i(t))$ is the maximum aposteriori estimator of $\bm_i^* = \pi^*\big( \bs(t) \big)$ given the local observation $\ro_i(t)$. 

Once the clustering in \eqref{eq: ABSA-2 - comms} is done, each agent $i$ will train its local communication policy $\pi^{c,ABSA-2}_i(\cdot)$, which is any non-injective mapping such that $\forall k \in \{1,...,|\mathcal{C}|\}: \pi^{c,ABSA-2}_i(\bo_i) = \bc^{(k)}$ iff $\bo_i \in \mathcal{P}_{i,k}$.
\textcolor{Mycolor3}{ After obtaining the communication policies $ \langle\pi^{c,ABSA-2}_i(\cdot) \rangle_{i=1}^N$, to obtain a proper control $\pi^m(\cdot)$ policy at the CC corresponding to the communication policies, we perform a single-agent reinforcement learning. To this end and to manage the complexity of the algorithm for larger values of $d$, we propose to use DQN architecture \cite{mnih2015human} at the CC.}

\vspace{-3mm}
\section{Performance Evaluation} \label{Numerical results - Section}
\vspace{-2mm}
In this section, we evaluate our proposed schemes via numerical results for the popular multi-agent geometric consensus problem\footnote{In our numerical experiments, the discount factor is assumed to be $\gamma = 0.9$. All experiments are done over a grid world of size $8 \times 8$, where the goal point of the rendezvous is located at the grid number $\Omega^T = \{ 22 \}$.}. Through indirect design, ABSA-1 and ABSA-2 never rely on explicit domain knowledge about any specific task, such as geometric consensus. Thus, we conjecture that their indirect design allows them to be applied beyond geometric consensus problems and to a much wider range of tasks. To make the geometric consensus task suitable for the evaluation of our proposed algorithms, similar to \cite{mostaani2022task}, we have introduced a bit constraint to the communication channel between the agents and the CC. After evaluating the proposed algorithms in the context of the rendezvous problem, we attempt to explain the behaviour of all the algorithms via the existing metric - positive listening - for measuring the task-effectiveness of communications. As positive listening falls short in explaining all the aspects of the behaviour of the investigated algorithms, we will also introduce a new metric. Called \textit{task relative information}, the new metric assists to further explain the behaviour of different algorithms with a higher accuracy and reliability.

\subsection{The geometric consensus problem}
\textcolor{Mycolor3}{ Our proposed schemes are evaluated in this section through numerical results for the rendezvous problem \cite{zilber2001communication,amato2009incremental}, which is a specific type of geometric consensus problems under finite observability \cite{barel2017come}. Following the instantaneous and synchronous communication model and the star network topology explained in section \ref{subsect: system model} and Fig. \ref{fig: communication topology} respectively, the rendezvous problem is explained as the following. At each time step $t$ several events happen in the following order. First, an agent $i$ obtains a local observation $\ro_i(t)$ - which is equivalent to its own location in the grid-world. The agent $i$, subsequently, follows its quantization/communication policy to generate a compressed version $\rc_i(t)$ of its observation to be communicated to the CC via bit-budgeted communication links. After receiving the quantized observations of all agents, CC follows its control policy to decide and select the joint action vector $\rm(t)$ and communicate each agent $i$'s local action $\rm_i(t)$ to it accordingly. The local action $\rm_i(t) \in \mathcal{M}$ that is communicated back to the agent $i$ via a perfect communication channel is a one directional move in the greed world, i.e, $\mathcal{M} = \{ \text{ left, right, up, down, pause} \}$. Given each agent $i$'s action $\rm_i(t)$ the environment evolves and transitions to the next time step $t+1$ where each agent $i$ obtains a new local observation $\ro_i(t+1)$. All agents receive a single team reward \begin{equation}\label{example environment noise distribution}
   \mathrm{r}_t =
   \begin{cases}
    C_1, & \text{if  $\exists \,  i,j \in N :  {\ro}_{i}(t) \in \Omega^T \And  {\ro}_{j}(t) \notin \Omega^T$} \\
    C_2, & \text{if  \textcolor{Mycolor3}{$\nexists \,  i \in N :  {\ro}_{i}(t) \in \Omega - \Omega^T$}},\\
    0, & \text{otherwise},\\
   \end{cases}
\end{equation}
where $C_1 < C_2$ and $\Omega^T$ is the set of terminal observations i.e., the episode terminates if $\exists \, i \in \mathcal{N}:  \ro_i(t) \in \Omega^T$. Accordingly, when not all agents arrive at the target point, a smaller reward $C_1 = 1$ is obtained, while the larger reward $C_2 = 10$ is attained when all agents visit the goal point at the same time.} We compare our proposed ABSA algorithms with the heuristic non-communicative (HNC), heuristic optimal communication (HOC) and SAIC algorithms proposed in \cite{mostaani2022task} which are direct schemes to jointly design the communication and control policies for the specific geometric consensus problem solved here. In contrast to ABSA-1 and ABSA-2 which enjoy an indirect design, the direct design of HOC and HNC does not allow them to be applied in any other problem rather than the specific geometric consensus problem with the finite observability i.e., the rendezvous problem explained here. 


\subsection{Numerical experiment}
 A constant learning rate $\alpha=0.07$ is applied when exact Q-learning is used to obtain $\pi^*(\cdot)$ and $\alpha=0.0007$ when DQN is used to learn $\pi^m(\cdot)$ for ABSA-2. For the exact Q-learning, a UCB\footnote{UCB is a standard scheme used in exact reinforcement learning to strike a trade-off between the exploration and exploitation \cite{Suttonintroduction}.} exploration rate of $c=1.25$ considered. The deep neural network that approximates the Q-values is considered to be a fully connected feed-forward network with 10 layers of depth, which is optimized using the Adam optimizer. An experience reply buffer of size 10'000 is used with the mini-batch size of 62. The target Q-network is updated every 1000 steps and for the exploration, decaying $\epsilon$-greedy with the initial $\epsilon = 0.05$ and final $\epsilon=0.005$ is used \cite{mnih2015human}. In any figure that the performance of each scheme is reported in terms of the averaged discounted cumulative rewards, the attained rewards throughout training iterations are smoothed using a moving average filter of memory equal to 20,000 iterations. As explained in section \ref{subsect: ABSA-1}, ABSA-1 and ABSA-2 both require a centralized training phase prior to be capable of being executed in a distributed fashion.


For all black curves, one prior centralized training phase to obtain $\pi^*(\cdot)$ is required. As detailed in Section III, the proposed algorithms, ABSA-1 and ABSA-2, leverage $\pi^*(\cdot)$ to design $\pi^c$ and then $\pi^m$ afterwards. Dashed curves, HOC and HNC, as proposed by \cite{mostaani2022task} provide heuristic schemes which exploit the domain knowledge of its designer about the rendezvous task making it not applicable to any other task rather than the rendezvous problem. While HOC enjoys a joint control and communication design, HNC runs with no communication. Note that HNC \& HOC require communication/coordination between agents prior to the starting point of the task - which is not required for any other scheme. These schemes, introduced by \cite{mostaani2022task}, are detailed as the following.

\begin{itemize}
    \item  A joint communication and control policy is designed \textbf {using domain knowledge} in the rendezvous problem. HNC agents approach the goal point and wait nearby for a sufficient number of time steps to ensure that the other agent has also arrived. Only after that, they will get to the goal point. Note that this scheme requires communication/coordination between agents prior to the starting point of the task, since they have to have had agreed upon this scheme of coordination.
     \item A joint communication and control policy is designed \textbf {using domain knowledge} in the rendezvous problem. HOC agents wait next to the goal point until the other agent informs them that they have also arrived there. Only after that, they will get to the goal point. Note that this scheme requires communication/coordination between agents prior to the starting point of the task, since they have to have had agreed upon this scheme of coordination and communications as well as on the the meaning that each communication message entails.
     \end{itemize}

\textcolor{Mycolor3}{To obtain the results demonstrated in Fig. \ref{fig: numerical result - return VS episode}, we have simulated the rendezvous problem for a three-agent system. The black curves illustrate the training phase that is occurring at CC to obtain $\pi^m$ after $\pi^c$ is already computed using equations (\ref{eq: ABSA-1}) and (\ref{eq: ABSA-2 - comms}).} We observe the lossless performance of ABSA-1 in achieving the optimal average return without requiring any \textcolor{Mycolor3}{(2nd round)} training. To enable fully decentralized quantization of the observation process, ABSA-2 was proposed which is seen to approach the optimal solution as $d$ grows. All ABSA-2 curves are plotted with $|\mathcal{C}| = 3$, and ABSA-1 curve is plotted with $|\mathcal{C}| = |\mathcal{M}|^N = 125$ in 3 agent scenarios - Fig. \ref{fig: numerical result - return VS episode} - and $ |\mathcal{C}| = |\mathcal{M}|^N = 25$ in the two agent scenario - Fig. \ref{fig: numerical result - Normalized return VS symnom}.

 \begin{figure}[t] 
  \centering 
      \includegraphics[width=0.98\textwidth]{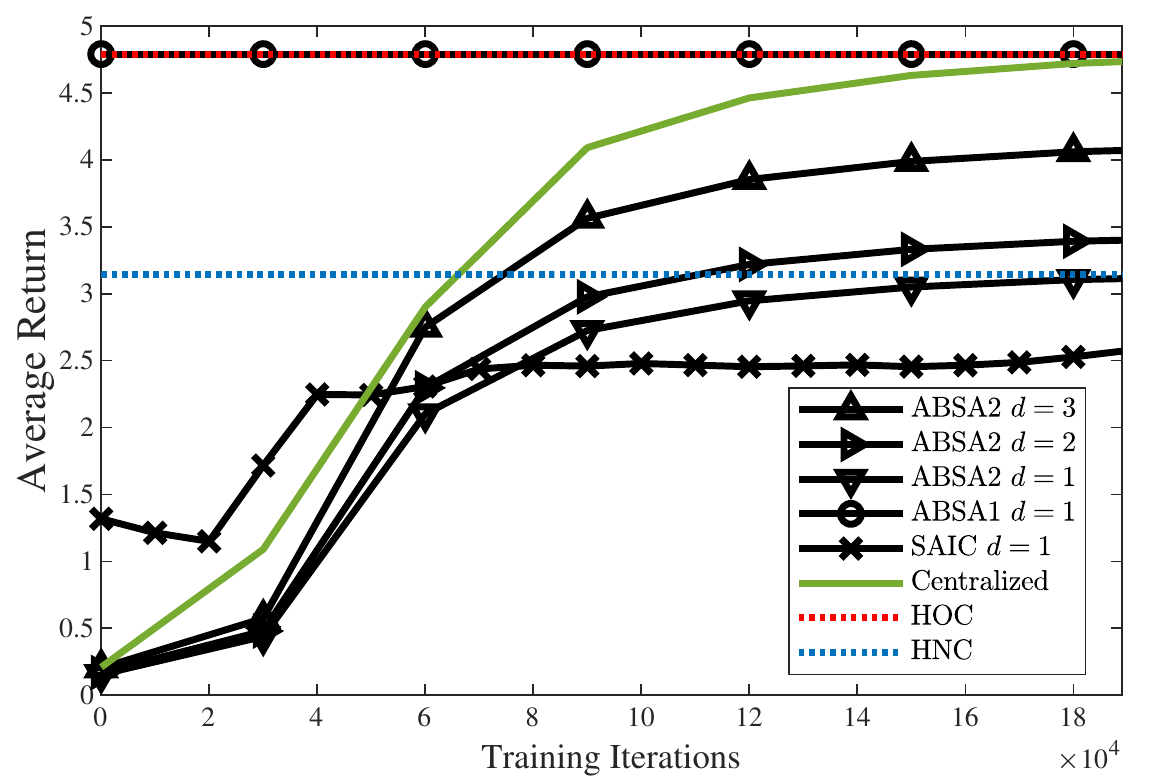} 
      \vspace{-1mm}
  \caption{ Average return comparison made between the proposed schemes and some benchmarks introduced in \cite{mostaani2022task} - the three agent scenario under constant bit-budget values.}
  \label{fig: numerical result - return VS episode}
  \vspace{-0.1cm}
\end{figure}

In Fig. \ref{fig: numerical result - return VS episode}, we see how the performance of ABSA-2 compares with HNC, HOC and SAIC at different rates of quantization. As expected, with the increase in the size of the quantization codebook, the average return performance of ABSA-2 is gradually improved, such that it approaches near-optimal performance at $d=3$. We also observe the superior performance of ABSA-2 compared with SAIC at very tight bit-budgets where SAIC's performance sees a drastic drop. It is observed that as $d$ grows, ABSA-2 approaches the optimal return performance even under higher rates of quantization, however, higher values of $d$ come at the cost of the increased computational complexity of ABSA-2.

 \begin{figure}[t] 
  \centering 
      \includegraphics[width=0.95\textwidth]{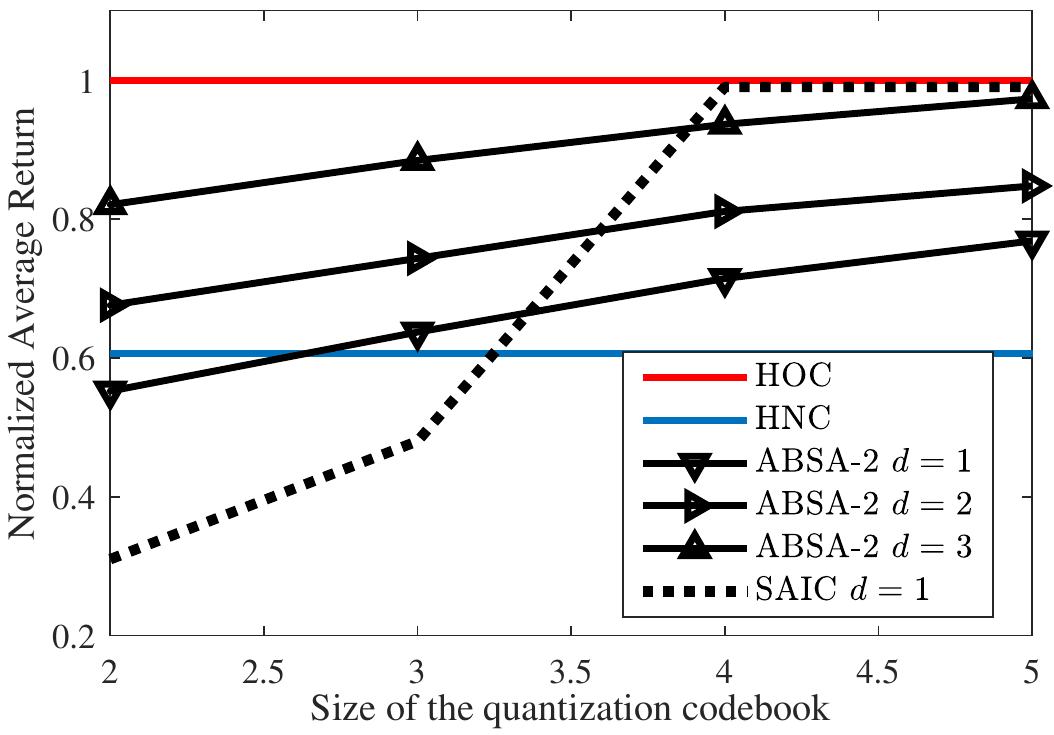} 
      \vspace{-1mm}
  \caption{ The obtained normalized average return as a function of codebook size $|\mathcal{C}|$ is compared across a range of schemes: proposed schemes and some benchmarks introduced in \cite{mostaani2022task} - two-agent scenario.}
  \label{fig: numerical result - Normalized return VS symnom}
  \vspace{-0.1cm}
\end{figure}


 \subsection{Explainablity of the learned communication policies}

One common metric to evaluate the effectiveness of communications in the literature \cite{lowe2019pitfalls} is \emph{positive listening} $I\big( \bc_i(t); \bm_j(t) \big) \,\, j \in \mathcal{N} - \{i\}$, which is the mutual information between the communication $\bc_i(t)$ produced by an agent $i$ and the action $\bm_j(t)$ selected by another agent following the receipt of the communication $\bc_i(t)$ from agent $i$. Positive signaling $I\big( \bo_i(t); \bc_i(t) \big)$ is another metric proposed by \cite{lowe2019pitfalls}, measuring the mutual information between agent $i$'s observation $\bo_i(t)$ and its own produced communication message $\bc_i(t)$ at the same time step. 
As to be shown below, however, these metrics are unable to fully capture the underlying performance trends of all schemes. Therefore, we, for the first time, introduce  a new metric called \emph{task relevant information} (RI) - allowing us to explain the task-effectiveness of the learned communication policies. 
 
 Measuring positive listening is one way to quantify the contribution of the communicated messages of agent $i$ to the action selection of agent $j$. Positive signalling, on the other hand, measures the consistency as well as the relevance of the communicated messages $\bc_i(t)$ and the agent's observations $\bo_i(t)$.
 As SAIC and ABSA use a deterministic mapping of observation $\ro_i$ to produce the communication message $\rc_i$, they are always guaranteed to have positive signalling \cite{lowe2019pitfalls} - the degree of which,  however, is limited by the uplink channel's bit budget $R = \log_2|\mathcal{C}|$. Thus, among the existing metrics for the measurement of the effectiveness of communications, we limit our numerical studies to the measurement of positive listening. It is known that the higher positive listening is, the stronger (not necessarily better) we expect the coordination between the agents to be. That is, the higher positive listening means higher degree of dependence between agents (their actions and observations) which is not necessarily sufficient for the team agents to fulfill the task.
 
 
 \begin{figure}[t] 
  \centering 
      \includegraphics[width=0.99\textwidth]{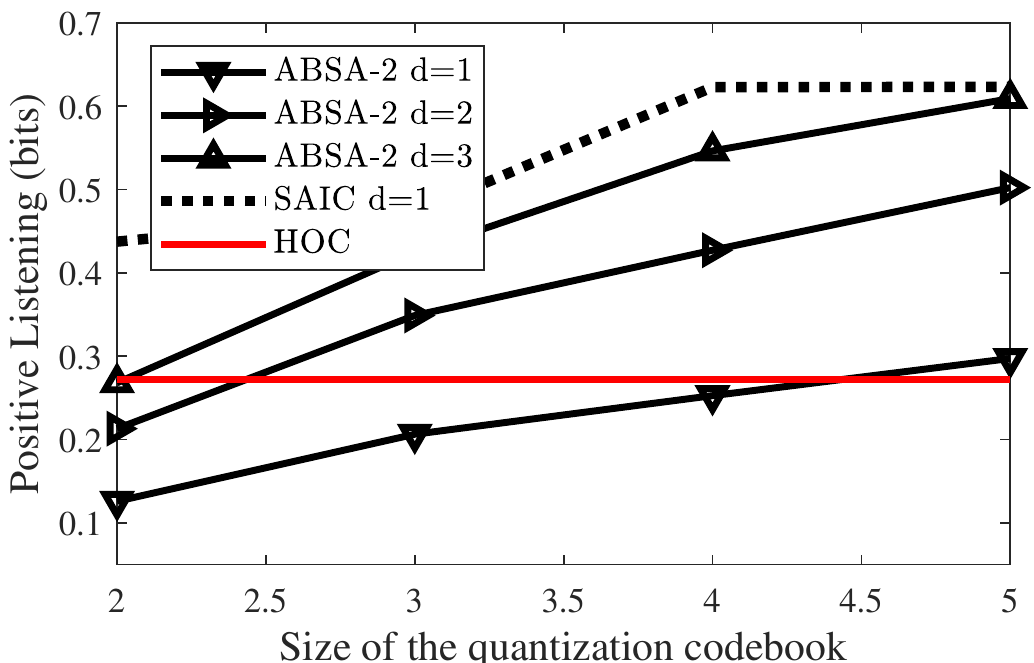} 
      \vspace{-1mm}
  \caption{Comparing the positive listening $I\big( \bc_i(t); \bm_j(t) \big)$ performance across a range of schemes.}
  \label{fig: numerical result - PL VS symnom}
  \vspace{-0.1cm}
\end{figure}

Figure \ref{fig: numerical result - PL VS symnom} explains how stronger coordination between agents and the CC is often resulting in an increased performance of the MAS in obtaining a higher average return. For instance, the enhancement in the positive-listening performance of SAIC from $|\mathcal{C}|=3$ to $|\mathcal{C}|=4$ quantizer in Fig. \ref{fig: numerical result - PL VS symnom} is resulting in an improved average return performance, as shown in Fig. \ref{fig: numerical result - Normalized return VS symnom}. This metric also reasonably explains the enhancement of ABSA-2 performance in obtaining higher return by increasing $d$ - the memory of the CC - and the size of the quantization codebook $|\mathcal{C}|$. Moreover, stronger coordination between agents and CC is visible in ABSA-2 when compared with HOC. Thus, we expect better average return performance for ABSA-2 which is in contrast to the results of Fig. \ref{fig: numerical result - return VS episode}. This event suggests that stronger coordination - measured by positive listening - may not necessarily result in an improved average return performance as the coordination may not be perfectly aligned with task needs. 

 The curve concerning the HOC scheme allows us to recall that a positive listening of 0.3 (bit) is sufficient to maintain the coordination required for optimal performance in the aforementioned geometric consensus task. Therefore, in the ABSA-2 and SAIC schemes, there is still an unnecessary influence from the side of the communication messages to the actions selected by the receiving end. In fact, not all the information received from the receiving end has contributed to the higher average return of the system. Accordingly, there is yet, some unnecessary data in the communication messages designed by ABSA that contain no task-specific/useful information. 
 
 Thus we believe that positive listening cannot explicitly quantify the effectiveness of the task-oriented communication algorithms; therefore they fall short in explaining the behaviour of these algorithms. Even when positive listening is computed as $I \left( \bc_i(t) ; \bm(t) \right)$ to capture the mutual information between the communication of agent $i$ and the control signals of all agents we arrive at almost similar patterns - Fig. \ref{fig: numerical result - PL vector VS symnom}. 
 
  \begin{figure}[t] 
  \centering 
      \includegraphics[width=0.99\textwidth]{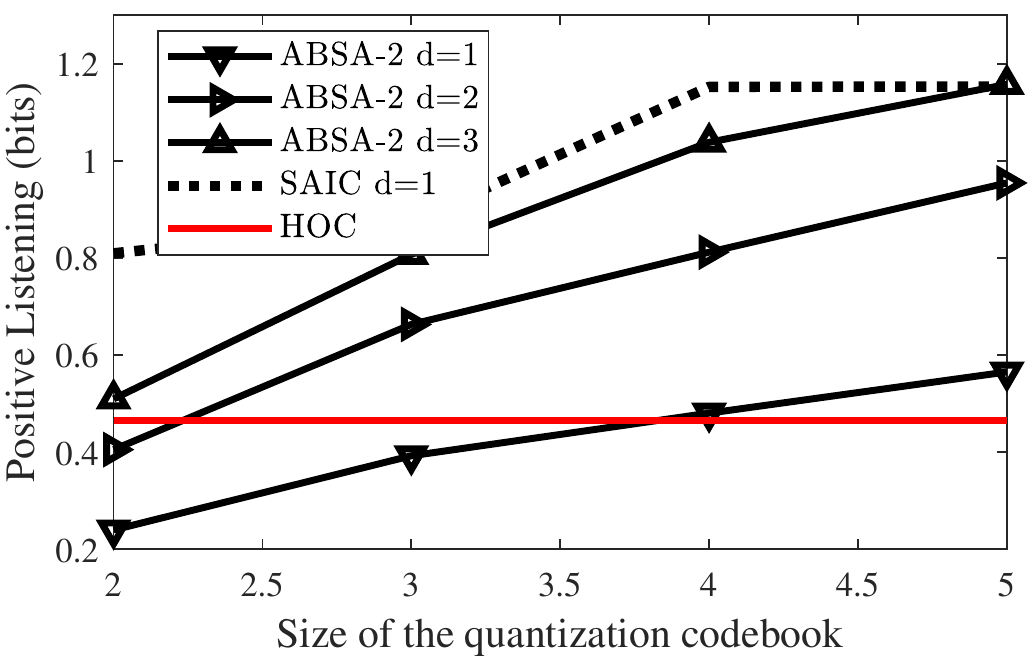} 
      \vspace{-1mm}
  \caption{Comparing the positive listening $I \left( \bc_i(t) ; \bm(t) \right)$ performance across a range of schemes.}
  \label{fig: numerical result - PL vector VS symnom}
  \vspace{-0.1cm}
\end{figure}

 Figure \ref{fig: numerical result - relevant information VS symnom}, investigates the performance of multiple schemes via a novel performance metric: task relevant information (TRI). Here we define the task relevant information metric to be
 \begin{align}
     I\Big( \pi^c\big(\bo_i(t)\big);  \pi^*\big( \bs(t) \big)  \Big) = I\big(\bc_i(t); \bm^*(t)\big) ,
 \end{align}
which measures the mutual information (in bits) between the communicated message of agent $i$ and the vector $\bm^*(t)$ of joint optimal actions at the CC - which is selected by the optimal centralized control policy $\pi^*(\cdot)$. As demonstrated by Fig. \ref{fig: numerical result - relevant information VS symnom}, TRI is an indirect metric of the effectiveness of communications that can explain the behaviour of different communication designs. It is also observed that the TRI metric magnifies the performance gap between different schemes as they get closer to the optimal performance. Nevertheless, TRI can be utilized as a standalone measure to quantify the effectiveness of a communication design since it almost perfectly predicts the average return performance of the a communication policy - without the need for the communication to be tested when solving the real task.

Note that, we measure the task-effectiveness of a quantization algorithm based on the average return that can be obtained when using it. Further, to measure the average return that can be obtained under the communication policies $\langle \pi^c_1(\cdot),..., \pi^c_N(\cdot) \rangle$, we have to design the control policy $\pi^m(\cdot)$ at the CC that selects the control vector $\bm(t)$ having access to only the quantized observations of the agents $\bc(t)$. Accordingly, we cannot measure the effectiveness of the communication policy of an MAS without having a specific design for their control policy. Even after the design of the control policy of the MAS, it is challenging to understand if the suboptimal performance of the algorithm is caused by an ineffective design of the control policy or the communication policy. In fact, it is hard disentangle the effect of the control and communication policies on the MAS's average return. Our proposed metric TRI can facilitate measuring the performance of any communication policy in isolation and without the effect of the control policy being present in the numerical values of TRI.

 Accordingly, the importance of introducing this metric is multi-fold: (i) by using TRI as an indirect metric we can measure the effectiveness of a communication policy for any specific task; (ii) it allows us to measure the effectiveness of the communication scheme prior to the design of any control policy; (iii) it helps to design task effective communication policies in complete separation from the control policy design.

  \begin{figure}[t] 
  \centering 
      \includegraphics[width=0.98\textwidth]{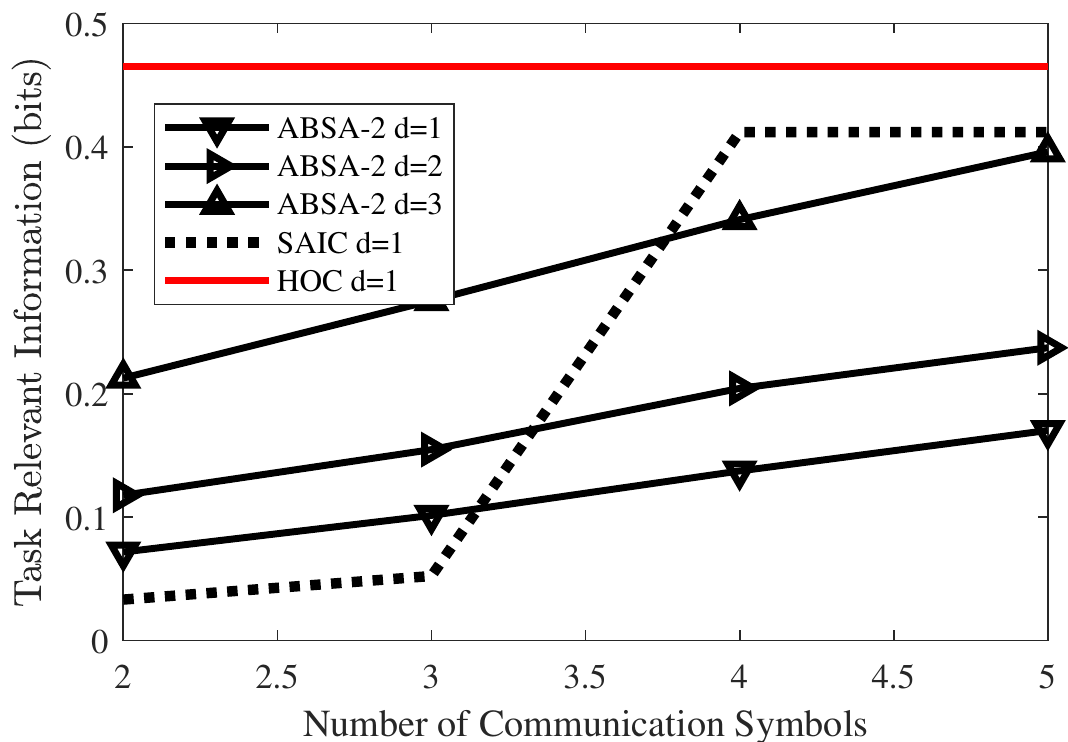} 
      \vspace{-1mm}
  \caption{Comparing the task relevant information (TRI) performance across a range of schemes. It is observed that TRI can comprehensively explain the behaviour of all task-effective quantization schemes in a certain task without the need to measure their effectiveness via their resulting average return in the task - compare this figure with Fig. \ref{fig: numerical result - Normalized return VS symnom} .}
  \label{fig: numerical result - relevant information VS symnom}
  \vspace{-0.1cm}
\end{figure}

\vspace{-4mm}
\section{Conclusion}
\vspace{-3mm}
In this paper, we have investigated the joint design of control and communications in an MAS under centralized control and distributed communication policies. We first proposed an action-based state aggregation algorithm (ABSA-1) for lossless compression and provided analytical proof of its optimality. Then we proposed ABSA-2, which offers a fully distributed communication policy and can trade computational complexity for communication efficiency. We finally demonstrated the task-effectiveness of the proposed algorithms via numerical experiments performed on a geometric consensus problem via a number of representative metrics. \textcolor{Mycolor3}{ 
 Furthermore, our numerical studies demonstrate the pressing need for further research on finding a metric that can measure/explain the task-effectiveness of communications with more accuracy.
And, scalability in task-oriented design is yet another central challenge to be addressed in future research.}
\appendices

\vspace{-4mm}
\section{Proof of Lemma \ref{lemma: quantization measure}} \label{append: proof of quantization measure}
\vspace{-2mm}
\label{append: proof of theorem: ABSA-1 is lossless}

\begin{proof}
Applying Adam's law on equation (\ref{eq: Joint Control and Communication Design (JCCD) problem}) yields
{\small
\begin{align}\label{eq: TBIC problem - adam's law applied}
& \underset{\pi}{\text{argmax }} 
& & \mathbb{E}_{p(\rc(t))}                   \Big{\{}
    \mathbb{E}_{p_{\pi^c, {\pi}^m}(\{\rtr\}_{t'}^{T'}|\bc(t))}                   \big{\{} 
   \bg(t') | \bc(t) 
\big{\}}   
\Big{\}}, ~ \text{s.t.} ~  
    |\mathcal{C}|\leq 2^R
\end{align}}
where $\bc(t)$ is generated by the communication policy $\pi^c$ and the joint pmf of the system's trajectory $\{tr\}_{t'}^{T'}$ is directly influenced by the action policy $\pi^m$. The conditional pmf $p_{\pi^c, {\pi}^m}(\{\rtr\}_{t'}^{T'}|\bc(t))$ is the joint probability of the trajectory of the system given the received communication $\bc(t)$ when policies $\pi^c(\cdot)$ and $ {\pi}^m(\cdot)$ are followed. We proceed by negating the equation (\ref{eq: TBIC problem - adam's law applied}) and adding a second term to the objective function which is constant with respect to the decision variables of the problem to have
{\small
\begin{align}\label{eq: TBIC problem - adam's law applied - 2}
& \underset{\pi^c}{\text{argmin }} 
& & \mathbb{E}_{p(\rs(t))}                   \Big{\{}
    \mathbb{E}_{p_{\pi^*}(\{\rtr\}_{t'}^{T'}|\rs(t))}                   \big{\{} 
   \bg(t') | \bs(t) 
\big{\}}   
\Big{\}}
- \\
&&& \mathbb{E}_{p(\rc(t))}                   \Big{\{}
    \mathbb{E}_{p_{\pi^c, {\pi}^m}(\{\rtr\}_{t'}^{T'}|\bc(t))}                   \big{\{} 
   \bg(t') | \bc(t) 
\big{\}}   
\Big{\}} , ~ \text{s.t.} ~  
    |\mathcal{C}|\leq 2^R. \notag
\end{align}}
We replace the conditional expectation of system return by the value function $V(\cdot)$, \cite{Suttonintroduction}(Ch. 3.5), and we will have
{\small
\begin{align}\label{eq: TBIC problem - adam's law applied - 3}
& \underset{\pi^c}{\text{argmin }} 
& & \mathbb{E}_{p(\rs(t))}                   \Big{\{}
    V^{\pi^*}\big(\bs(t) \big) 
\Big{\}}
- \mathbb{E}_{p(\rc(t))}                   \Big{\{}
    V^{ {\pi}^m}\big(\bc(t) \big)  
\Big{\}} ,
\notag\\
& \text{s.t.} & &  
    |\mathcal{C}|\leq 2^R.
\end{align}}
Note that the empirical joint distribution of $ \bc(t)$ can be obtained by following the communication policy $\pi^c$ on the empirical distribution of $\bs(t)$.
{\small
\begin{align}\label{eq: TBIC problem - adam's law applied - 3}
& \underset{\pi^c}{\text{argmin }} 
& & \mathbb{E}_{p(\rs(t))}                   \Big{\{}
    V^{\pi^*}\big(\bs(t) \big) 
\Big{\}}
- \mathbb{E}_{p(\rs(t))}                   \Big{\{}
    V^{ {\pi}^m}\big(\bc(t) \big)  
\Big{\}} ,
\notag\\
& \text{s.t.} & &  
    |\mathcal{C}|\leq 2^R.
\end{align}}
As $    V^{\pi^*}\big(\bs(t) \big) 
-                 
V^{ {\pi}^m}\big(\bc(t) \big)  \geq 0
$ is true for any $\bs(t) \in \mathcal{S}$, merging the two expectations results in
{\small
\begin{align}\label{eq: TBIC problem - adam's law applied - 4}
& \underset{\pi^c}{\text{argmin }} 
& & \mathbb{E}_{p(\rs(t))}                   \Big{|}
    V^{\pi^*}\big(\bs(t) \big) 
-                
    V^{ {\pi}^m}\big(\bc(t) \big)  
\Big{|} , ~~ \text{s.t.} ~~ |\mathcal{C}|\leq 2^R,
\end{align}}
which concludes the proof of the lemma.
\end{proof}

\section{Proof of Lemma \ref{Theorem: ABSA-1 is Lossless}} \label{append: proof: Theorem: ABSA-1 is Lossless}
\begin{proof}
We depart from the result of lemma \ref{lemma: quantization measure} - problem (\ref{eq: TBIC problem - adam's law applied - 4 - body}).  By taking the expectation over the empirical distribution of $\bs(t)$ and applying Bellman optimality equation, we obtain
{\small
\begin{align}\label{eq: TBIC problem - adam's law applied - 6}
& \underset{\pi}{\text{argmin }} 
& & \frac{1}{n} \sum_{t=1}^{n}                  \Big{|}
    Q^{\pi^*}\!\big(\bs(t), \pi^*(\bs(t)) \big)\!-\!            
    Q^{ {\pi}^m}\!\Big(\bc(t),  {\pi}^m\!\big( \pi^c(\bs(t))\big)\! \Big)\!  
\Big{|} ,
\notag\\
& \text{s.t.} & &  
    |\mathcal{C}|\leq 2^R, 
\end{align}}
where the vector $\pi^c(\bs(t))$ is of $N$ dimensions and its $i$-th element is $\bc_i(t)$. We proceed by plugging $\pi^{c,ABSA-1}(\cdot)$ and $\Pi^{m^*}$, according to the definition \ref{def: ABSA-1}, into the equation (\ref{eq: TBIC problem - adam's law applied - 6}) to obtain
\vspace{-1mm}
{\small
\begin{align}\label{eq: TBIC problem - adam's law applied - 7}
\frac{1}{n} \sum_{t=1}^{n}                  \Big{|}
    Q^{\pi^*}\big(\bs(t), \pi^*(\bs(t)) \big) 
-                
    Q^{\pi^*}\Big(\bc(t), {\pi}^*\big( \bs' \big) \Big)  
\Big{|} , 
\end{align}}
where $\bs' = \pi^{{c,ABSA-1}^{-1}}\Big(\pi^{c,ABSA-1}\big(\bs(t)\big)\Big)$, and any possible value for it lies in the same subset $  \mathcal{P}_{k'}$ as $\bs(t)$ does, while according to the definition of $ \mathcal{P}_{k'}$, we know $ \pi^*(\bs(t)) =  \pi^*(\bs')$, \textcolor{Mycolor3}{ if $|\mathcal{C}| \geq |\mathcal{M}|^N$}. Thus, by replacing $\pi^*(\bs')$ in with $\pi^*(\bs(t))$ in equation (\ref{eq: TBIC problem - adam's law applied - 7}) we get
{\small
\begin{align}\label{eq: TBIC problem - adam's law applied - 7}
\frac{1}{n} \sum_{t=1}^{n}                  \Big{|}
    Q^{\pi^*}\big(\bs(t), \pi^*(\bs(t)) \big) 
-                
    Q^{\pi^*}\Big(\bs(t), {\pi}^*\big( \bs(t) \big) \Big)  
\Big{|} = 0.
\end{align}}
This concludes the proof of theorem \ref{Theorem: ABSA-1 is Lossless}.

\end{proof}

\bibliographystyle{IEEEtran}
\vspace{-4mm}
{\small
\bibliography{Bibfile}}
\vspace{-0mm}
\end{document}